\documentclass[letterpaper, 10 pt, conference]{ieeeconf}
\IEEEoverridecommandlockouts %
\usepackage{mypreamble}

\title{\bf \LARGE Computing the Hard Scaled Relative Graph of LTI Systems}
\author{Julius P.~J. Krebbekx$^1$, Eder Baron-Prada$^{2,3}$, Roland Tóth$^{1,4}$, Amritam Das$^1$
\thanks{$^1$Control Systems group, Department of Electrical Engineering,  Eindhoven University of Technology, The Netherlands.}
\thanks{$^2$Austrian Institute of Technology, 1210 Vienna,
Austria.}
\thanks{$^3$Automatic Control Laboratory, ETH Zurich, 8092 Z\"urich, Switzerland.}
\thanks{$^4$Systems and Control Lab, HUN-REN Institute for Computer Science and Control, Budapest, Hungary. 
E-mail: {\tt\small \{j.p.j.krebbekx, r.toth, am.das\}@tue.nl, ebaron@ethz.ch}}
}

\date{\today}

\begin{document}

\maketitle

\begin{abstract}
    \emph{Scaled Relative Graphs} (SRGs) provide a novel graphical frequency-domain method for the analysis of nonlinear systems, where \emph{Linear Time-Invariant} (LTI) systems are the fundamental building block. To analyze feedback loops with unstable LTI components, the \emph{hard} SRG is required, since it aptly captures the input/output behavior on the extended $L_2$ space. In this paper, we develop a systematic computational method to \emph{exactly} compute the hard SRG of LTI systems, which may be unstable and contain integrators. We also study its connection to the Nyquist criterion, including the multivariable case, and demonstrate our method on several examples.
\end{abstract}

\section{Introduction}

For \emph{Linear Time-Invariant} (LTI) systems, graphical system analysis using the Nyquist diagram~\cite{nyquistRegenerationTheory1932} is the cornerstone of control engineering. It is easy to use, allows for intuitive analysis and controller design methods and has been widely used in the industry. However, it is unclear how to systematically generalize graphical frequency domain methods to nonlinear system analysis and controller design.

The \emph{Scaled Relative Graph} (SRG)~\cite{ryuScaledRelativeGraphs2022} was recently proposed as a novel graphical framework for analyzing nonlinear \emph{Single-Input Single-Output} (SISO) systems~\cite{chaffeyGraphicalNonlinearSystem2023}. The SRG offers a non-approximative method yielding sufficient conditions for stability and upper bounds on the (incremental) $L_2$-gain, a key performance metric in practice. It is modular, allowing interconnections to be analyzed by composing the SRGs of subsystems. The SRG recovers classical results such as the small-gain theorem~\cite{chaffeyGraphicalNonlinearSystem2023} and generalizes the circle criterion~\cite{krebbekxScaledRelativeGraph2025}, due to its close relation to the Nyquist diagram. It also enables frequency-dependent gain bounds, forming the basis of a nonlinear Bode diagram and bandwidth definition~\cite{krebbekxNonlinearBandwidthBode2025}, accommodates for the LTI notions of phase lead/lag~\cite{eijndenPhaseScaledGraphs2025}, and can be used to analyze nonlinear \emph{Multi-Input Multi-Output} (MIMO) systems~\cite{krebbekxGraphicalAnalysisNonlinear2025}. Apart from theoretical developments, SRGs have proven effective in applications such as reset control analysis~\cite{grootDissipativityFrameworkConstructing2025} and design~\cite{krebbekxResetControllerAnalysis2025}, and the analysis of decentralized systems~\cite{baron-pradaDecentralizedStabilityConditions2025}.

The central building block in SRG analysis of nonlinear systems is the SRG of LTI systems. For stable SISO systems, the SRG was computed in~\cite{chaffeyGraphicalNonlinearSystem2023}, and it was shown in~\cite{krebbekxScaledRelativeGraph2025a} that this description was problematic for unstable LTI systems in the loop. The resolution provided in~\cite{krebbekxScaledRelativeGraph2025a} was to add the Nyquist criterion to the SRG, obtaining the so-called \emph{extended} SRG, solving the problem of how to deal with unstable systems in the loop for the SISO case. 

The \emph{hard} SRG was defined in~\cite{chenSoftHardScaled2025}, where the original SRG from~\cite{chaffeyGraphicalNonlinearSystem2023} was called the \emph{soft} SRG. As argued in~\cite{chenSoftHardScaled2025}, the hard SRG can be computed for unstable systems, hence provides an SRG framework that can deal with unstable LTI systems in the loop. An initial attempt to bound the hard SRG of stable square LTI systems is made in~\cite{grootDissipativityFrameworkConstructing2025}. However, a systematic method to compute the hard SRG of LTI systems, both stable and \emph{unstable} (including integrators), is lacking. Since the benefit of hard SRGs is mostly to deal with unstable systems in a feedback loop, a method to compute the hard SRG of unstable LTI systems would present a significant advancement for SRG analysis of nonlinear systems.

In this paper, we develop a non-approximative method to compute the hard SRG of LTI systems, which may be unstable. Our method is based on the transfer function representation of LTI systems, and constructs a bound from intersecting and removing disks in the complex plane centered on the real axis, and we prove that this bound is exact. This type of bound has been used in \cite{grootDissipativityFrameworkConstructing2025} for bounding the soft and hard SRG of stable square LTI and reset systems, and in \cite{krebbekxGraphicalAnalysisNonlinear2025} to bound the soft SRG of stable LTI operators which need not be square. We also prove that the extended SRG~\cite{krebbekxScaledRelativeGraph2025a}, which combines the SISO Nyquist criterion and the soft SRG, is \emph{identical} to the hard SRG. Having the hard SRG of LTI systems available, we show how they can serve as an alternative to the MIMO Nyquist criterion~\cite{macfarlaneGeneralizedNyquistStability1977} and the Generalized Nyquist Criterion~\cite{desoerGeneralizedNyquistStability1980} (GNC).

This paper is structured as follows. We present the mathematical preliminaries, such as the various SRG definitions, in Section~\ref{sec:preliminaries}. In Section~\ref{sec:hard_srg_lti} we present our main result, the algorithm for computing the hard SRG of LTI systems. Section~\ref{sec:connection_nyquist_criterion} discusses the connection to the Nyquist criterion, both in the SISO and MIMO case. We demonstrate our method on several examples in Section~\ref{sec:examples}, and present our conclusions in Section~\ref{sec:conclusion}.

\section{Preliminaries}\label{sec:preliminaries}

\subsection{Notation and Conventions}

Let $\R, \C$ denote the real and complex number fields, respectively, with $\R_{\geq 0} = [0, \infty)$. Let $\C_\infty := \C \cup \{ \infty \}$ denote the extended complex plane. A closed disk in $\C_\infty$ centered at $x$ with radius $r$ is defined as $D_r(x) = \{ z \in \C \mid |z-x| \leq r \}$, where the unit disk is $\mathbb{D} = D_1(0)$, and we define $D_\infty(x) = \C_\infty$. We denote the field of rational functions in $s \in \C$ with real coefficients as $\R(s)$~\cite{desoerFeedbackSystemsInputoutput1975}. For $G(s) \in \R^{q \times p}(s)$, we define $\overline{\sigma}(G) := \sup_{\omega \in \R} \overline{\sigma}(G(j \omega))$ and $\underline{\sigma}(G) := \inf_{\omega \in \R} \underline{\sigma}(G(j \omega))$, where $\overline{\sigma}(A)$ ($\underline{\sigma}(A)$) denotes the largest (smallest) singular value of the matrix $A \in \C^{q \times p}$. The distance between two sets $\mathcal{C}_1,\mathcal{C}_2 \subseteq \C_\infty$ is defined as $\dist(\mathcal{C}_1,\mathcal{C}_2) := \inf_{z_1 \in \mathcal{C}_1, z_2 \in \mathcal{C}_2} |z_1-z_2|$, where $|\infty-\infty|:=0$. Inversion $\mathcal{C}^{-1}$ of some $\mathcal{C} \in \C_\infty$ is defined as the M\"obius inverse $r e^{j \phi} \mapsto (1/r)e^{j \phi}$. The complement of a set $\mathcal{C}$ in its ambient space (e.g. $\R,\C$ or $\C_\infty$) is denoted as $\mathcal{C}^c$. The set difference between $\mathcal{C}_1,\mathcal{C}_2$ is defined as $\mathcal{C}_1 \setminus \mathcal{C}_2 =\{z \in \mathcal{C}_1 \mid z \notin \mathcal{C}_2 \}$.

\subsection{Signals and Systems}

\subsubsection{Signal spaces}

Let $\mathcal{L}$ denote a Hilbert space, equipped with an inner product $\inner{\cdot}{\cdot}_\mathcal{L} : \mathcal{L} \times \mathcal{L} \to \C$ and norm $\norm{x}_\mathcal{L} := \sqrt{\inner{x}{x}_\mathcal{L}}$. For $d \in \N^+$, $\mathbb{T} \in \{ [0,T], \R_{\geq 0} \mid T>0 \}$, the Hilbert spaces of interest are
\begin{equation}\label{eq:L2d-space}
    L_2^d(\mathbb{T}) := \{ f: \mathbb{T} \to \R^d \mid \norm{f} <\infty \},
\end{equation}
with inner product $\inner{f}{g}:= \int_\mathbb{T} f(t)g(t) d t$ inducing the norm $\norm{f}$. For $\R^d$, the inner product is $xy = \sum^d_{i=1} x_i y_i$. We will abbreviate $L_2^d(\R_{\geq 0}) =: L_2^d$ and $L_2^d([0,T]) =: L_2^d[0,T]$.

Elements of $f \in L_2^d(\mathbb{T})$ are denoted as $f=(f_1,\dots,f_d)^\top$, where $f_i \in L_2(\mathbb{T})$ for all $i=1,\dots,d$. Furthermore, $0 \in L_2^d(\mathbb{T})$ refers to the map $\mathbb{T} \ni t \mapsto 0 \in \R^d$. 

For any $T \in \R_{\geq 0}$, define the truncation operator $P_T : L_2(\mathbb{F}) \to L_2(\mathbb{F})$ as $(P_T u)(t) = 0$ for all $t>T$, else $(P_T u)(t) = u(t)$,
which we abbreviate as $u_T :=P_T u$. The extension of $L_2^d$, see Ref.~\cite{desoerFeedbackSystemsInputoutput1975}, is defined as 
\begin{equation*}
    \Lte^d := \{ u : \R_{\geq 0} \to \R^d \mid \norm{P_T u}_2 < \infty \text{ for all } T \in \R_{\geq 0} \},
\end{equation*}
The extended space is the natural setting for modeling systems, as it includes periodic and diverging signals, which are otherwise excluded from $L_2$.

For $n \leq d$, we define the linear subspaces (which are Banach spaces)
\begin{subequations}
\begin{align}
    \mathcal{U}^d_n &:= \{ f \in L_2^d \mid f_i = 0 \text{ for } i>n \}, \label{eq:subspace-Un-zeros} \\
    \mathcal{U}^d_{n, T} &:= \{ f \in L_2^d[0,T] \mid f_i = 0 \text{ for } i>n \}. \label{eq:subspace-UnT-zeros}
\end{align}
\end{subequations}
In the above definitions, the superscript $d$ is dropped if it is clear from the context.

\subsubsection{Systems}

Systems are modeled as operators $R: \Lte^p \to \Lte^q$, which means that $R$ maps \emph{each} element of $\Lte^p$ into $\Lte^q$. Such an operator is said to be causal (on $\Lte^p$) if it satisfies $P_T (R u) = P_T(R(P_Tu))$ for all $u \in \Lte^p$ and $T \in \R_{\geq 0}$~\cite{desoerFeedbackSystemsInputoutput1975}. Unless otherwise specified, \emph{we will always assume causality.} A causal operator $R: \Lte^p \to \Lte^q$ can be truncated to $[0,T]$ by defining $R|_T : L_2^p[0,T] \to L_2^q[0,T]$ as $R|_T u_T := (R u_T)_T$ for any $T>0$ and identifying a $u_T \in \Lte^p$ uniquely with a $u_T \in L_2^p[0,T]$. 

For an operator $R: X \to Y$, where $(X,\norm{\cdot}_X)$ and $(Y,\norm{\cdot}_Y)$ are normed spaces, we define the induced incremental norm as (similar to the notation in~\cite{vanderschaftL2GainPassivityTechniques2017})
\begin{equation}\label{eq:incremental_induced_norm}
    \Gamma(R):= \sup_{x_1,x_2 \in X} \frac{\norm{R x_1-R x_2}_Y}{\norm{x_1-x_2}_X} \in [0,\infty].
\end{equation}

\subsection{LTI Systems}

We consider LTI systems $G : \Lte^p \to \Lte^q$ that can be represented by a finite-dimensional state-space representation $(A,B,C,D)$, where $A \in \R^{n_\mathrm{x} \times n_\mathrm{x}}, B \in \R^{n_\mathrm{x} \times p}, C \in \R^{q \times n_\mathrm{x}}, D \in \R^{q \times p}$, or equivalently by a real rational transfer function $G(s) = C (sI - A)^{-1} B+D \in \R^{q \times p}(s)$~\cite{hespanhaLinearSystemsTheory2009}. We denote the space of proper and stable real rational transfer function as $RH_\infty^{q \times p}$~\cite{zhouRobustOptimalControl1996}. For $R \in RH_\infty^{q \times p}$, then it is evident from Parseval's identity that $\Gamma(R) = \overline{\sigma}(R) = \norm{R}_\infty$~\cite{desoerFeedbackSystemsInputoutput1975}. 

\subsubsection{Transfer functions, poles and zeros}

A transfer function $G(s) \in \R^{q \times p}(s)$ is proper if $\lim_{s \to \infty} G(s) = D$ is bounded, otherwise $G(s)$ is called improper. If $D$ is non-singular, then $G(s)$ is called bi-proper. We call a pole $p_i$ stable if it lies in the left-half plane (LHP), i.e., $\mathrm{Re}(p_i) <0$, and $G(s)$ is stable if it is proper and has only stable poles. Similarly, a pole $p_i$ is unstable if it lies in the closed right half-plane (RHP), i.e., $\mathrm{Re}(p_i) \geq 0$, and if $G(s)$ has at least one unstable pole, it is called unstable. 

When all transmission zeros of $G(s)$, i.e., those $z_i \in \C$ such that $G(z_i)$ loses rank, lie in the open LHP ($\mathrm{Re}(z_i) <0$), then $G(s)$ is called \emph{minimum phase}.

\subsubsection{Inverting LTI systems}

A transfer function $G(s) \in \R^{q \times p}(s)$ is said to have full column normal rank if $G(s)$ has rank $p$ in $\R^{q \times p}(s)$. We note that the normal rank, i.e., the rank of $G(s)$ in $\R^{q \times p}(s)$, is equal to the number of nonzero entries in the Smith-McMillan form~\cite{hespanhaLinearSystemsTheory2009}.

An LTI system $G : \Lte^p \to \Lte^q$ can be inverted if and only if $p=q$, $G(s)$ is bi-proper and has full normal rank~\cite{hespanhaLinearSystemsTheory2009}. In that case, the poles of $G(s)$ become the zeros of $G^{-1}(s)$, and vice versa.

\subsection{Scaled Relative Graphs}

Since the introduction of the SRG, several extensions/adaptations have been proposed. We refer to the original SRG as the \emph{soft SRG}~\cite{ryuScaledRelativeGraphs2022}, since it considers signals in $L_2$. The soft SRG was extended in~\cite{krebbekxScaledRelativeGraph2025a,krebbekxScaledRelativeGraph2025} by adding the encirclement information from the Nyquist criterion, resulting in the \emph{extended SRG}. The \emph{hard} SRG, introduced in~\cite{chenSoftHardScaled2025}, considers trajectories in $L_2[0,T]$ for all $T>0$, hence inherts its name from the hard IQC~\cite{megretskiSystemAnalysisIntegral1997}.  

\subsubsection{Hyperbolic convexity} 

A set $\mathcal{C} \subset \C$ is hyperbolic convex (h-convex) if $\mathcal{C} = g(\mathrm{co}(f(\mathcal{C}))) =: \hco(\mathcal{C})$, where $f(z) = \frac{(\bar{z}-j)(z-j)}{1+|z|^2}$, $g(z) = \{\frac{\mathrm{Im}(z) \pm \sqrt{1-|z|^2}}{\mathrm{Re}(z)-1} \}$, and $\mathrm{co}$ ($\hco$) denotes the (hyperbolic) convex hull. It was shown in~\cite{patesScaledRelativeGraph2021} that the closure of every h-convex set $\mathcal{C}$ obeys
\begin{equation}\label{eq:h-convex-set-disk-representation}
    \mathrm{cl} \mathcal{C} = \bigcap_{\alpha \in \R} D_{R_\alpha}(\alpha) \setminus D_{r_\alpha}(\alpha)  =: \mathcal{G}(\{r_\alpha\}, \{R_\alpha\}),
\end{equation}
for some choice of radii $0\leq r_\alpha \leq R_\alpha \leq \infty$ for all $\alpha \in \R$. 

SRG-based stability analysis is based on the \emph{distance} between SRGs. Since taking the closure of sets does not affect their distance in $\C_\infty$, it suffices to identify an SRG with its closure.

\subsubsection{Soft SRGs}

Let $\mathcal{L}$ be a Hilbert space, and $R : \dom(R) \subseteq \mathcal{L} \to \mathcal{L}$ a relation. Define the angle between $u, y\in \mathcal{L}$ as $\angle(u, y) := \cos^{-1} \frac{\mathrm{Re} \inner{u}{y}}{\norm{u} \norm{y}} \in [0, \pi]$. Given distinct signals $u_1, u_2 \in \mathcal{U} \subseteq \dom(R)$, we define the set
\begin{multline*}
    z_R(u_1, u_2) := \\ \left\{ \frac{\norm{y_1-y_2}}{\norm{u_1-u_2}} e^{\pm j \angle(u_1-u_2, y_1-y_2)} \mid y_1 \in R u_1, y_2 \in R u_2 \right\} \\ \cup \{\infty \mid \text{if } R \text{ is multi-valued }\}.
\end{multline*}
The (soft) SRG of $R$ over the set $\mathcal{U}$ is defined as
\begin{equation}\label{eq:soft_srg_def}
    \SRG_\mathcal{U} (R) := \bigcup_{u_1, u_2 \in \mathcal{U}, \; u_1 \neq u_2} z_R(u_1, u_2) \subseteq \C_\infty,
\end{equation}
and we denote $\SRG(R) := \SRG_{\dom(R)}(R)$. Note that $0 \in \SRG_\mathcal{U}(R)$, if and only if there exist $u_1, u_2 \in \mathcal{U}$, $u_1 \neq u_2$, such that $R u_1 = R u_2$. We refer the reader to~\cite{krebbekxGraphicalAnalysisNonlinear2025} for the definition of the SRG for non-square MIMO operators $R : \Lte^p \to \Lte^q$, i.e., where $p \neq q$.

\subsubsection{Soft SRG of LTI systems}

In~\cite{patesScaledRelativeGraph2021} it was shown that $\SRG(T)$ is h-convex for any linear operator $T$ on a Hilbert space. For stable SISO LTI systems, i.e., $G \in RH_\infty$, it was shown that the soft SRG~\eqref{eq:soft_srg_def} is equal to the hyperbolic convex hull of the Nyquist diagram~\cite{patesScaledRelativeGraph2021,chaffeyGraphicalNonlinearSystem2023}, i.e.
\begin{equation}\label{eq:soft_srg_lti_siso}
    \SRG(G) = \hco(\{G(j \omega) \mid \omega \in \R \}) = \mathcal{G}(\{r_\alpha\}, \{R_\alpha\}).
\end{equation}
In~\cite{krebbekxGraphicalAnalysisNonlinear2025}, it was shown that $r_\alpha = \inf_{\omega \in \R} |G(s) - \alpha|$ and $R_\alpha = \sup_{\omega \in \R} |G(s) - \alpha|$ are the largest values of $r_\alpha$, and smallest values of $R_\alpha$, such that $\SRG(G) \subseteq \mathcal{G}(\{r_\alpha\}, \{R_\alpha\})$, and therefore~\eqref{eq:soft_srg_lti_siso} holds for these $r_\alpha, R_\alpha$ values.

If $G$ is not stable,~\eqref{eq:soft_srg_lti_siso} still holds, with the same values of $r_\alpha$ and $R_\alpha$, since stability of $G$ does not affect the hyperbolic convex hull computation in~\eqref{eq:soft_srg_lti_siso}. However, now $\SRG(G) = \SRG_\mathcal{U}(G)$, where $\mathcal{U} \subseteq L_2$ is the set of signals such that $G :\mathcal{U} \to L_2$~\cite{takedaInstabilityFeedbackSystems1973,krebbekxScaledRelativeGraph2025}. 

Analogously to~\cite{krebbekxGraphicalAnalysisNonlinear2025}, we embed MIMO operators $G$ in a space of square operators of size $n = \max\{p,q\}$ and define
\begin{equation}\label{eq:G_alpha}
    G_\alpha = \mat{G \\ 0_{(n-q) \times p}} - \mat{\alpha I \\ 0_{(n-p) \times p}},
\end{equation}
where $\alpha \in \R$. In this case, one has the bound
\begin{equation}\label{eq:soft_srg_lti}
    \SRG(G) \subseteq \mathcal{G}(\{\underline{\sigma}(G_\alpha)\}, \{\overline{\sigma}(G_\alpha)\}),
\end{equation}
as shown in~\cite{krebbekxGraphicalAnalysisNonlinear2025}. The bound in~\eqref{eq:soft_srg_lti} becomes an equality if $G(s)$ is a normal matrix for all $s \in j \R$~\cite{krebbekxGraphicalAnalysisNonlinear2025,grootDissipativityFrameworkConstructing2025}. In fact, we can prove the latter for all stable and square systems.

\begin{proposition}\label{prop:stable_square_lti_soft_srg}
    Let $G \in RH_\infty^{p \times p}$, then $\SRG(G) = \mathcal{G}(\{\underline{\sigma}(G_\alpha)\}, \{\overline{\sigma}(G_\alpha)\})$.
\end{proposition}

\begin{proof}
    See the appendix.
\end{proof}

\subsubsection{The extended SRG}

To extend SRG analysis beyond stable LTI systems, an \emph{extended} SRG for SISO LTI systems was defined in~\cite{krebbekxScaledRelativeGraph2025} by adding the Nyquist criterion information to the soft SRG. Consider the LTI feedback system in Fig.~\ref{fig:linear_feedback}, where $L(s)$ is a transfer function. The Nyquist criterion is given by the following theorem.

\begin{figure}[tb]
    \centering

    \tikzstyle{block} = [draw, rectangle, 
    minimum height=2em, minimum width=2em]
    \tikzstyle{sum} = [draw, circle, node distance={0.5cm and 0.5cm}]
    \tikzstyle{input} = [coordinate]
    \tikzstyle{output} = [coordinate]
    \tikzstyle{pinstyle} = [pin edge={to-,thin,black}]
    
    \begin{tikzpicture}[auto, node distance = {0.3cm and 0.5cm}]
        \node [input, name=input] {};
        \node [sum, right = of input] (sum) {$ $};
        \node [block, right = of sum] (lti) {$L(s)$};
        \node [coordinate, right = of lti] (z_intersection) {};
        \node [output, right = of z_intersection] (output) {}; %
        \node [coordinate, below = of lti] (static_nl) {};
    
        \draw [->] (input) -- node {$r$} (sum);
        \draw [->] (sum) -- node {$e$} (lti);
        \draw [->] (lti) -- node [name=z] {$y$} (output);
        \draw [-] (z) |- (static_nl);
        \draw [->] (static_nl) -| node[pos=0.99] {$-$} (sum);
    \end{tikzpicture}
    
    \caption{Feedback interconnection with LTI loop transfer $L(s)$.}
    \label{fig:linear_feedback}
\end{figure}
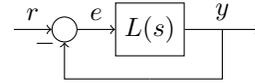

\begin{theorem}[Ch. 4.5~\cite{desoerFeedbackSystemsInputoutput1975}]\label{thm:nyquist}
    Let $n_\mathrm{z}$ denote the number of unstable closed-loop poles, and $n_\mathrm{p}$ the number of unstable poles of $L(s)$\footnote{Poles of $L(s)$ on the imaginary axis do not contribute to $n_\mathrm{p}$. The D-contour is chosen such that poles on the imaginary axis are kept to the left of the D-contour.}. Additionally, denote $n_\mathrm{n}$ the amount of times that $L(j \omega)$ encircles the point $-1$ in clockwise fashion as $\omega $ traverses the $D$-contour, going from $-jR$ to $jR$ and then along $Re^{j \phi}$ as $\phi$ goes from $\pi/2 \to -\pi/2$, for $R \to \infty$. The closed-loop system in Fig.~\ref{fig:linear_feedback} satisfies $n_\mathrm{z}=n_\mathrm{n}+n_\mathrm{p}$.
\end{theorem}

Note that for stability, $n_\mathrm{z}=0$ is required. Based on the Nyquist criterion, the extended SRG is defined as follows.

\begin{definition}\label{def:lti_srg_extended}
    Let $G$ be an LTI system with transfer function $G(s) \in \R(s)$ that has $n_\mathrm{p}$ unstable poles, then
    \begin{equation}\label{eq:set_of_encircled_unstable_points}
        \mathcal{N}_G = \{ z \in \C \mid N_G(z) +n_\mathrm{p} >0 \},
    \end{equation}
    where the winding number $N_G(z) \in \Z$ denotes the amount of clockwise encirclements of $z$ by $G(j \omega)$\footnote{See Theorem~\ref{thm:nyquist} for details on the Nyquist contour.}. The extended SRG of an LTI operator is defined as 
    \begin{equation}\label{eq:lti_srg_redefinition}
        \SRG'(G) := \SRG(G) \cup \mathcal{N}_G.
    \end{equation}
\end{definition}

\subsubsection{Hard SRGs}

The \emph{hard} SRG of $R: \Lte^p \to \Lte^p$, as defined in~\cite{chenSoftHardScaled2025}, is obtained by taking the union of $\frac{\norm{y_{1T}-y_{2T}}}{\norm{u_{1T}-u_{2T}}} e^{\pm j \angle(u_{1T}-u_{2T}, y_{1T}-y_{2T})}$ for all $T>0$, $u_1,u_2\in \Lte^p$ and $y_i=R u_i$ for $i\in \{1,2\}$. For causal $R$, this is by definition equivalent to 
\begin{equation}\label{eq:hard_srg_def}
    \SRGe(R) = \bigcup_{T>0} \SRG(R|_T),
\end{equation}
where it is understood that $\SRG(R|_T) = \SRG_{L_2^p[0,T]}(R|_T)$. 

Note that the hard SRG can be generalized to MIMO systems, whereas the extended SRG cannot since it is based on the SISO Nyquist criterion. On the other hand, as opposed to the extended SRG, the hard SRG does not have a direct interpretation in terms of the Nyquist criterion, and there is no systematic method available to compute the hard SRG for LTI systems that can handle stable and unstable systems.

\subsection{Stability Analysis using SRGs}\label{sec:srg_stab_analysis}

We will discuss how the three different SRGs (soft, hard, extended) can be used to analyze feedback systems. For simplicity, we focus on the stability of the LTI closed loop $R = (1+L^{-1})^{-1}$ in Fig.~\ref{fig:linear_feedback}, but much more complicated nonlinear feedback systems can be analyzed using SRGs~\cite{krebbekxScaledRelativeGraph2025}.

\subsubsection{Soft SRGs}

Stability analysis tools using the soft SRG were developed in~\cite{chaffeyGraphicalNonlinearSystem2023}, and require $L \in RH_\infty$. Then, Fig.~\ref{fig:linear_feedback} is stable if and only if $\dist(-\tau, \SRG(L)^{-1}) \geq r>0$ for some $r>0$ and for all $\tau \in [0,1]$. Moreover, $\Gamma(R) =1/r$.

\subsubsection{Extended SRGs}

The extended SRG allows one to discard the homotopy condition. The system in Fig.~\ref{fig:linear_feedback} is stable if and only if~\cite{krebbekxScaledRelativeGraph2025} $\dist(-1, \SRG'(L)^{-1}) \geq r' >0$, for some $r'>0$, and $\Gamma(R) = 1/r'$.

\subsubsection{Hard SRGs}\label{sec:hard_srg_stability}

Hard SRG stability analysis, as developed in~\cite{chenSoftHardScaled2025}, has the same advantage as the extended SRG, i.e., that the homotopy condition can be removed. The system $R$ is stable if $\dist(-1, \SRGe(L)^{-1}) \geq \tilde{r}>0$ for some $\tilde{r}>0$. Note that this condition shows that $\Gamma(R|_T) \leq 1/ \tilde{r}$ for all $T>0$, hence $\Gamma(R) \leq 1/\tilde{r}$ must hold.

\subsubsection{Comparison}\label{sec:srg_comparisons}

We note that both the soft SRG and the extended SRG, besides providing an incremental $L_2$-gain bound, also guarantee \emph{well-posedness} (see~\cite{megretskiSystemAnalysisIntegral1997}) of the feedback system~\cite{krebbekxScaledRelativeGraph2025a}. The extended SRG has the additional benefit that the homotopy condition need not be checked. Even though the hard SRG shares this benefit with the extended SRG, the hard SRG has the disadvantage that well-posedness must be \emph{assumed}.

\section{The Hard SRG of LTI Systems}\label{sec:hard_srg_lti}

Let $G : \Lte^p \to \Lte^q$ be an LTI system with real rational transfer function $G(s) \in \R^{q \times p}(s)$. The main result of this work is an algorithm to compute $\SRGe(G)$, and will the focus of this section.

\subsection{Bounding the Hard SRG with Circles}

The type of bound we consider is obtained by intersecting annuli in $\C_\infty$ that are centered on the real axis (see~\cite{patesScaledRelativeGraph2021}). These are disks $D_R(\alpha)$ with a smaller disk $D_r(\alpha)$ removed, where $\alpha \in \R$ and $0\leq r \leq R \leq \infty$. The following lemma shows that a bound for $\SRGe(G_\alpha)$ translates to a bound for $\SRGe(G)$ shifted by $\alpha$.

\begin{lemma}\label{lemma:G_alpha_bound}
    Let $G : \Lte^p \to \Lte^q$ and $\SRGe(G_\alpha) \subseteq \mathcal{C}$ for some $\alpha \in \R$ and $G_\alpha$ defined in \eqref{eq:G_alpha}, then
    \begin{equation}\label{eq:G_alpha_bound}
        \SRGe(G) \subseteq \alpha + \mathcal{C}.
    \end{equation}
\end{lemma}
\begin{proof}
    See the appendix.
\end{proof}

Now suppose that for all $u \in \Lte^p$ and $T>0$, it holds that 
\begin{equation}\label{eq:r_R_alpha_defs}
    0 \leq r_\alpha \leq \norm{(G_\alpha u)_T} / \norm{u_T} \leq R_\alpha \leq \infty,
\end{equation}
where $r_\alpha$ ($R_\alpha$) is the largest (smallest) value such that~\eqref{eq:r_R_alpha_defs} holds. Then, by the definition of the hard SRG~\eqref{eq:hard_srg_def}, and linearity of $G_\alpha$, we can conclude $\SRGe(G_\alpha) \subseteq D_{R_\alpha}(0) \setminus D_{r_\alpha}(0)$, which is, by Lemma~\ref{lemma:G_alpha_bound}, equivalent to
\begin{equation}\label{eq:hard_srg_G_circle_bound}
    \SRGe(G) \subseteq D_{R_\alpha}(\alpha) \setminus D_{r_\alpha}(\alpha).
\end{equation}
In general,~\eqref{eq:hard_srg_G_circle_bound} will not be an equality. Instead, we are interested in a bound for the hard SRG in the form of~\eqref{eq:soft_srg_lti}, i.e., intersecting~\eqref{eq:hard_srg_G_circle_bound} for $\alpha \in \R$.

\subsection{Computing the Maximum Gain Radius $R_\alpha$}

The smallest \emph{upper bound} $R_\alpha$ for $G_\alpha$ in \eqref{eq:r_R_alpha_defs} is found using the following proposition.

\begin{proposition}\label{prop:maxgain}
    Let $G : \Lte^p \to \Lte^q$ be an LTI system with transfer function $G(s) \in \R^{q \times p}(s)$, then $R_\alpha$ in~\eqref{eq:r_R_alpha_defs} reads
    \begin{equation}
        R_\alpha = \overline{\sigma}(G_\alpha) \text{ if } G \in RH_\infty^{q \times p}, \text{ else } R_\alpha = \infty.
    \end{equation}
\end{proposition}
\begin{proof}
    See the appendix.
\end{proof}

\subsection{Computing the Minimum Gain Radius $r_\alpha$}

Whereas $R_\alpha$ could be easily computed via the $H_\infty$-norm, computing the smallest lower bound $r_\alpha$ is more challenging. The difficulty arises since~\eqref{eq:r_R_alpha_defs} must hold for all $T>0$, and using $\underline{\sigma}(G_\alpha)$ via Parseval's identity only provides a valid lower bound if $G_\alpha$ is stable and $T \to \infty$.

The solution is to \emph{invert} a square $G_\alpha$, and apply Proposition~\ref{prop:maxgain} to find a maximum gain on the inverse $u = G_\alpha^{-1} y_T$. %

\begin{proposition}\label{prop:mingain_square}
    Let $G : \Lte^p \to \Lte^p$ be a square LTI system with transfer function $G(s) \in \R^{p \times p}(s)$, then $r_\alpha$ in~\eqref{eq:r_R_alpha_defs} reads
    \begin{equation}
        r_\alpha = \underline{\sigma}(G_\alpha) \text{ if } G_\alpha \text{ is minimum-phase, else } r_\alpha = 0.
    \end{equation}
\end{proposition}
\begin{proof}
    See the appendix.
\end{proof}

We note that Proposition~\ref{prop:mingain_square} could be extended to the case where $G_\alpha : \Lte^p \to \Lte^p$ is non-square. In the wide case ($p>q$), the transfer function $G_\alpha(s)$ will always have a nontrivial algebraic kernel, i.e., $G_\alpha(s)U(s)=0$ for some $U(s) \in \R^p(s)$, which implies $r_\alpha=0$. In the tall case ($p<q$), one can use a left pseudo-inverse for $G_\alpha(s)$. The proof for the tall case requires a significant amount of additional technical details, hence we only consider the square case in this paper. 

Another reason for studying only the square case is that a MIMO LTI feedback system can often be written as a square system $L$ in simple feedback interconnection such as in Fig.~\ref{fig:linear_feedback}, since the analysis of the loop transfer boils down to $L=GK$, which is square.

Finally, we note that it is known that $r_\alpha = \underline{\sigma}(G)$ and $R_\alpha = \overline{\sigma}(G)$ for stable systems satisfies~\eqref{eq:r_R_alpha_defs} for $T \to \infty$~\cite{desoerFeedbackSystemsInputoutput1975,zhouRobustOptimalControl1996}, and that the $R_\alpha$ result for any $T > 0$ can be generalized to the nonlinear case (e.g., see~\cite[Ch. 1]{vanderschaftL2GainPassivityTechniques2017}). Propositions~\ref{prop:maxgain} and~\ref{prop:mingain_square} provide extensions to finite $T>0$ and unstable systems.

\subsection{Overview of the Bounding Algorithm}

By combining Propositions~\ref{prop:maxgain} and~\ref{prop:mingain_square}, we can formulate the following theorem that \emph{exactly} computes the hard SRG of a square LTI system.

\begin{theorem}\label{thm:hard_srg_bound_lti_circles}
    Let $G : \Lte^p \to \Lte^p$ be a square LTI system with transfer function $G(s) \in \R^{p \times p}(s)$, then 
    \begin{equation}\label{eq:hard_srg_G_circle_bound_thm}
        \SRGe(G) = \mathcal{G}(\{r_\alpha\}, \{R_\alpha \}),
    \end{equation}
    where $r_\alpha, R_\alpha$ are obtained using Propositions~\ref{prop:maxgain} and~\ref{prop:mingain_square}.
\end{theorem}
\begin{proof}
    See the appendix.
\end{proof}

Note that~\eqref{eq:hard_srg_G_circle_bound_thm} becomes an inclusion ($\subseteq$) if one takes $\alpha \in \mathcal{A} \subsetneq \R$. Therefore, numerical implementations of Theorem~\ref{thm:hard_srg_bound_lti_circles} will never \emph{under}-approximate $\SRGe(G)$. 

If $G(s) \notin \R^{p \times p}(s)$ since it contains a time delay, $r_\alpha =0$ always holds, and $R_\alpha$ is unaffected.

\section{Connection to the Nyquist Criterion}\label{sec:connection_nyquist_criterion}

\subsection{SISO Nyquist Criterion}

In the SISO case, i.e., $p=1$, the hard SRG formula from Theorem~\ref{thm:hard_srg_bound_lti_circles} has a direct connection to the Nyquist criterion and the extended SRG from~\cite{krebbekxGraphicalAnalysisNonlinear2025}. 

\begin{theorem}\label{thm:hard_and_extended_srg_equivalence}
    Let $G : \Lte \to \Lte$ be an LTI system with transfer function $G(s) \in \R(s)$, then $\SRGe(G) = \SRG'(G)$.
\end{theorem} %

\begin{proof}
    See the appendix.
\end{proof}

An immediate consequence of Theorem~\ref{thm:hard_and_extended_srg_equivalence} is that in the SISO case, analysis of feedback systems with the hard SRG \emph{does} guarantee well-posedness.

\subsection{MIMO Nyquist Criterion}\label{sec:mimo_nyquist}

We will now discuss how hard SRG stability analysis in Section~\ref{sec:hard_srg_stability} can be used as an alternative to the Nyquist criterion for MIMO LTI feedback systems~\cite{macfarlaneGeneralizedNyquistStability1977}, and the GNC~\cite{desoerGeneralizedNyquistStability1980}.

\subsubsection{An alternative Nyquist criterion}

Let $L : \Lte^p \to \Lte^p$ be a square LTI system with transfer function $L(s) \in \R^{p \times p}(s)$. Then, stability of the closed loop $R$ in Fig.~\ref{fig:linear_feedback} can be directly checked by verifying if $\dist(-1, \SRGe(L)^{-1}) \geq r >0$, which is equivalent to $\dist(-1, \SRGe(L)) \geq r' >0$. 

Because of~\eqref{eq:hard_srg_def}, the SRG interconnection rules on Hilbert spaces, see e.g.~\cite{ryuScaledRelativeGraphs2022,krebbekxGraphicalAnalysisNonlinear2025}, carry over to the hard SRG.
In particular, $\SRGe(k L) = \SRGe(L) k = k \SRGe(L)$ for $0 \neq k \in \R$ and $\SRGe(I+L)=1+\SRGe(L)$, where $I$ is the identity matrix in $\C^{p \times p}$. Therefore, one can use $\SRGe(L)$ to design gains $k, k_1,k_2 \in \R$ such that
\begin{subequations}\label{eq:stability_criteria}
\begin{align}
    \dist(-1/k, \SRGe(L)) & \geq r'>0, \text{ or } \label{eq:k_criterion} \\ \dist(-1, k_1+k_2 \SRGe(L)) & \geq r'' >0. \label{eq:k1k2_criterion}
\end{align}
\end{subequations}
The gain $k$ would be placed in the feedback path, and the gains $k_1,k_2$ define the new loop transfer $L' = k_1 I +k_2 L$, and these choices of gains yield a stable feedback system if the corresponding stability criterion in~\eqref{eq:stability_criteria} is satisfied. Note that the MIMO SRG interconnection rules~\cite{krebbekxGraphicalAnalysisNonlinear2025} only work for scalars, so when $k,k_1,k_2$ were to be \emph{matrices} of gains, one must directly compute $\SRGe(k L)$ or $\SRGe(k_1 + k_2 L)$.

The hard SRG also naturally provides a stability margin for the system in Fig.~\ref{fig:linear_feedback}, which is the smallest distance between $-1$ and $\SRGe(L)^{-1}$. If the gains $k,k_1,k_2$ above are varied, the result on stability and performance can be directly checked without having to re-compute the hard SRG.

If one merely wants to check stability of a feedback loop, Theorem~\ref{thm:hard_srg_bound_lti_circles} can be used by evaluating only $r_{-1}$, i.e., $\alpha=-1$. If $r_{-1}>0$, then one is assured that $\dist(-1, \SRGe(L)) \geq r_{-1}>0$, which proves stability. Since the sensitivity reads $S=(1+L)^{-1}$, the bound $\norm{S}_\infty = \Gamma(S) \leq 1/r_{-1}$ holds.

\subsubsection{Comparison with MIMO Nyquist and GNC}

As opposed to the MIMO Nyquist criterion, one does not have to count encirclements to assess stability using the hard SRG. Additionally, for the MIMO Nyquist criterion, one has to re-compute the encirclements of $\det(I+L)$ \emph{every time} the gains are varied, whereas $\SRGe(L)$ can be re-used in calculations with different gains.

The GNC aimed to eliminate this problem of having to re-compute the winding number upon a change of gain. However, the GNC requires gluing the eigenloci by hand, which can behave numerically poor, especially when integrators are involved. Additionally, the GNC does not directly provide an $L_2$-gain bound, and no robustness margin.

We note, however, that Theorem~\ref{thm:hard_srg_bound_lti_circles} currently relies on a model for the LTI system. The (MIMO) Nyquist criterion is particularly relevant since one can replace the Nyquist curve by frequency response measurement data, hence allows for design and analysis using a plant non-parametric model. To truly compare the hard SRG with the (MIMO) Nyquist criterion, the radii $r_\alpha, R_\alpha$ in Theorem~\ref{thm:hard_srg_bound_lti_circles} should be computed from input/output or FRF data directly. In particular, one should be able to detect non-minimum phase zeros and unstable poles from input/output data. 

Note that for Proposition~\ref{prop:maxgain}, one only needs to know \emph{if} there are unstable poles, but not \emph{how many} there are, as opposed to the MIMO Nyquist criterion and GNC. The detection of a non-minimum phase zero from data, however, is known to be challenging~\cite{bolandPhaseResponseReconstruction2021}.

\begin{figure*}[t]
     \centering
     \begin{subfigure}[b]{0.32\linewidth}
         \centering
         \includegraphics[width=\linewidth]{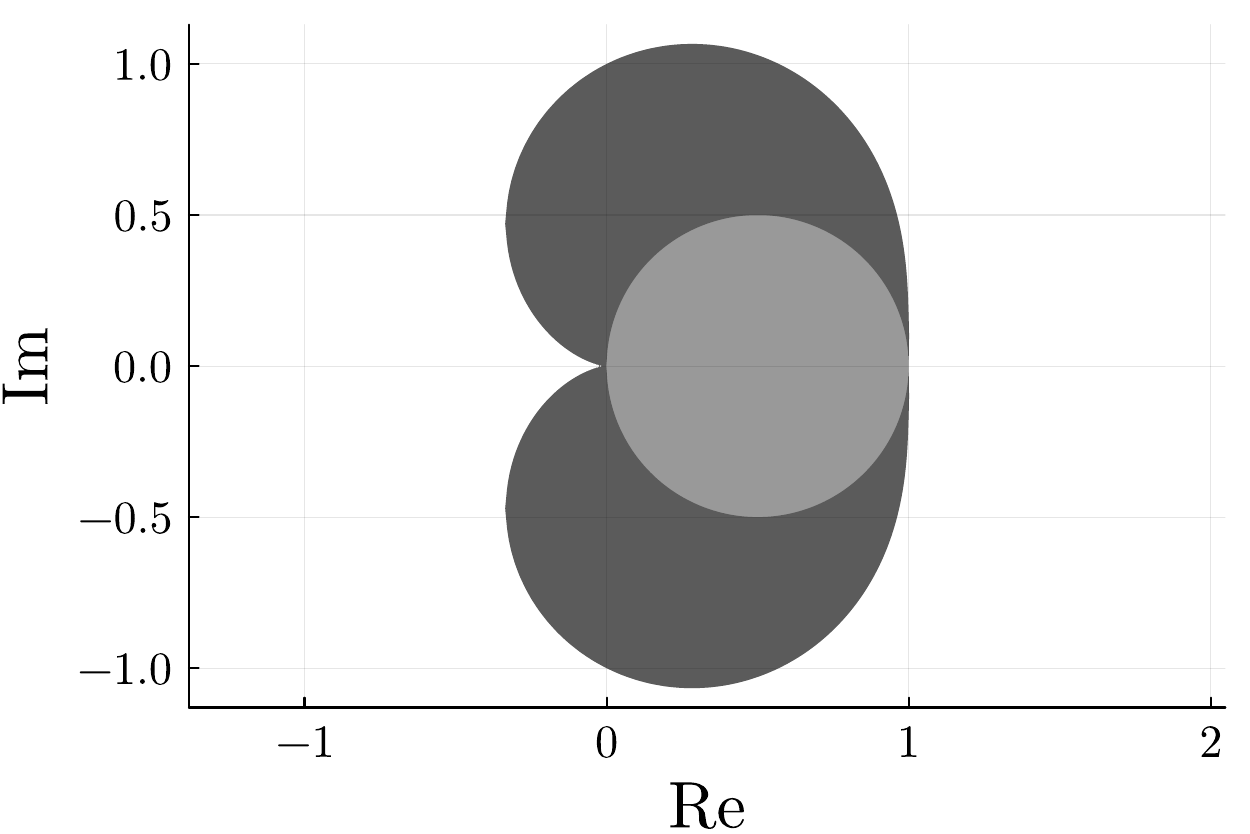}
         \caption{$\SRGe(G_1)$.}
         \label{fig:example_siso_1}
     \end{subfigure}
     \hfill
     \begin{subfigure}[b]{0.32\linewidth}
         \centering
         \includegraphics[width=\linewidth]{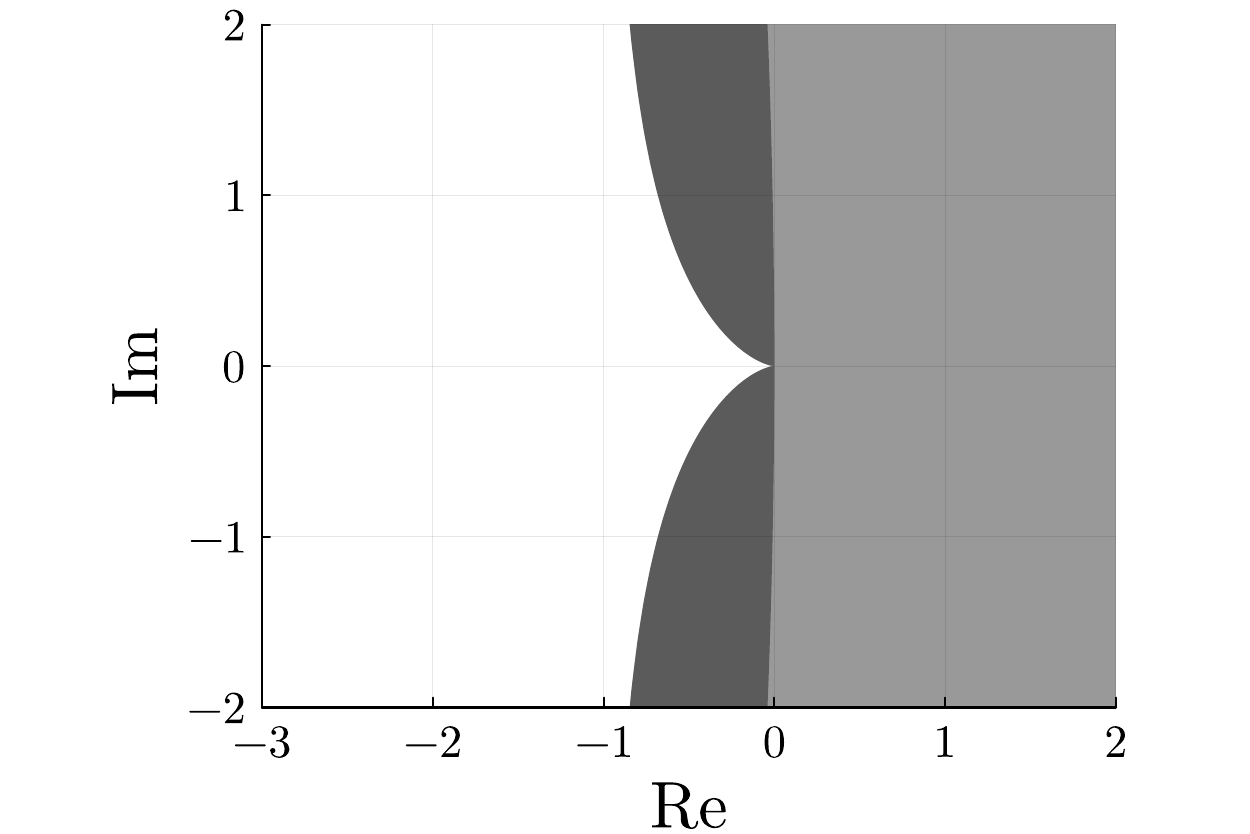}
         \caption{$\SRGe(G_2)$.}
         \label{fig:example_siso_2}
     \end{subfigure}
     \hfill
     \begin{subfigure}[b]{0.32\linewidth}
         \centering
         \includegraphics[width=\linewidth]{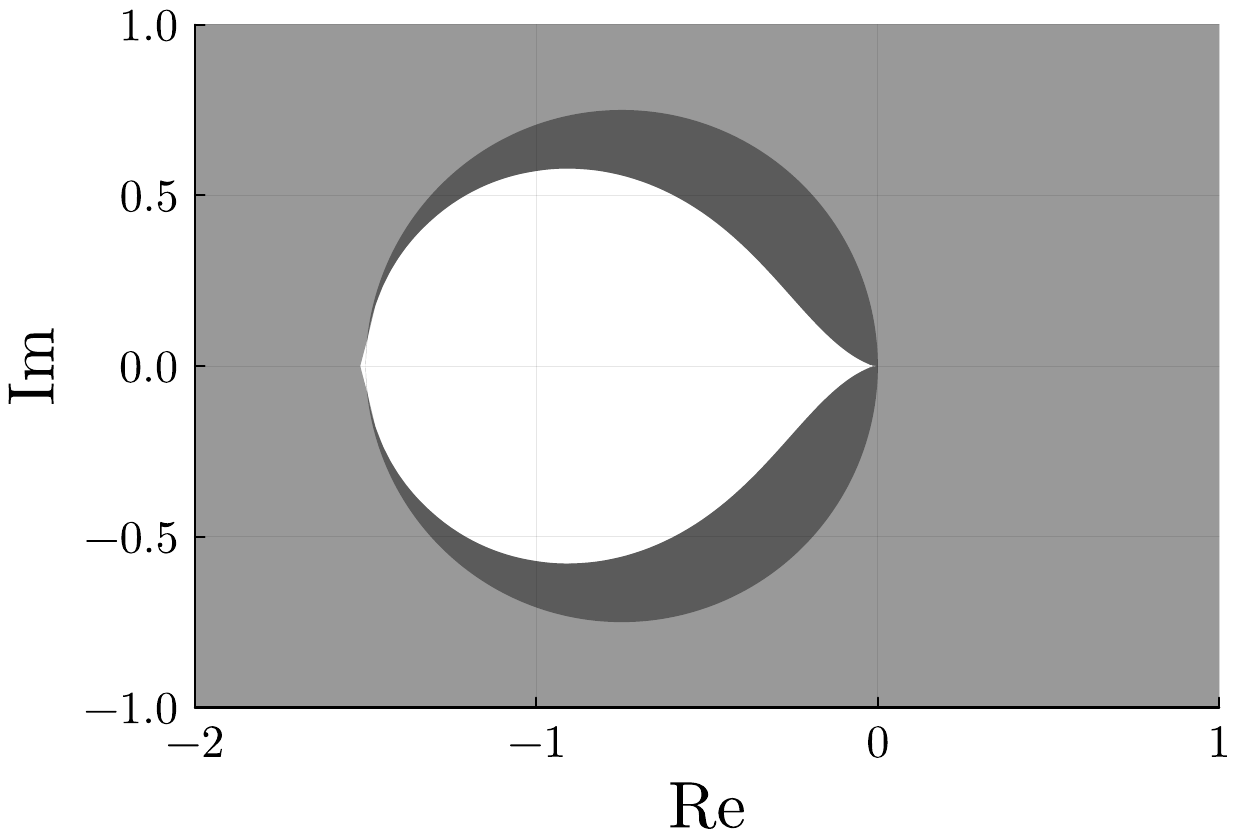}
         \caption{$\SRGe(G_3)$.}
         \label{fig:example_siso_3}
     \end{subfigure}
     \hfill
     \begin{subfigure}[b]{0.32\linewidth}
         \centering
         \includegraphics[width=\linewidth]{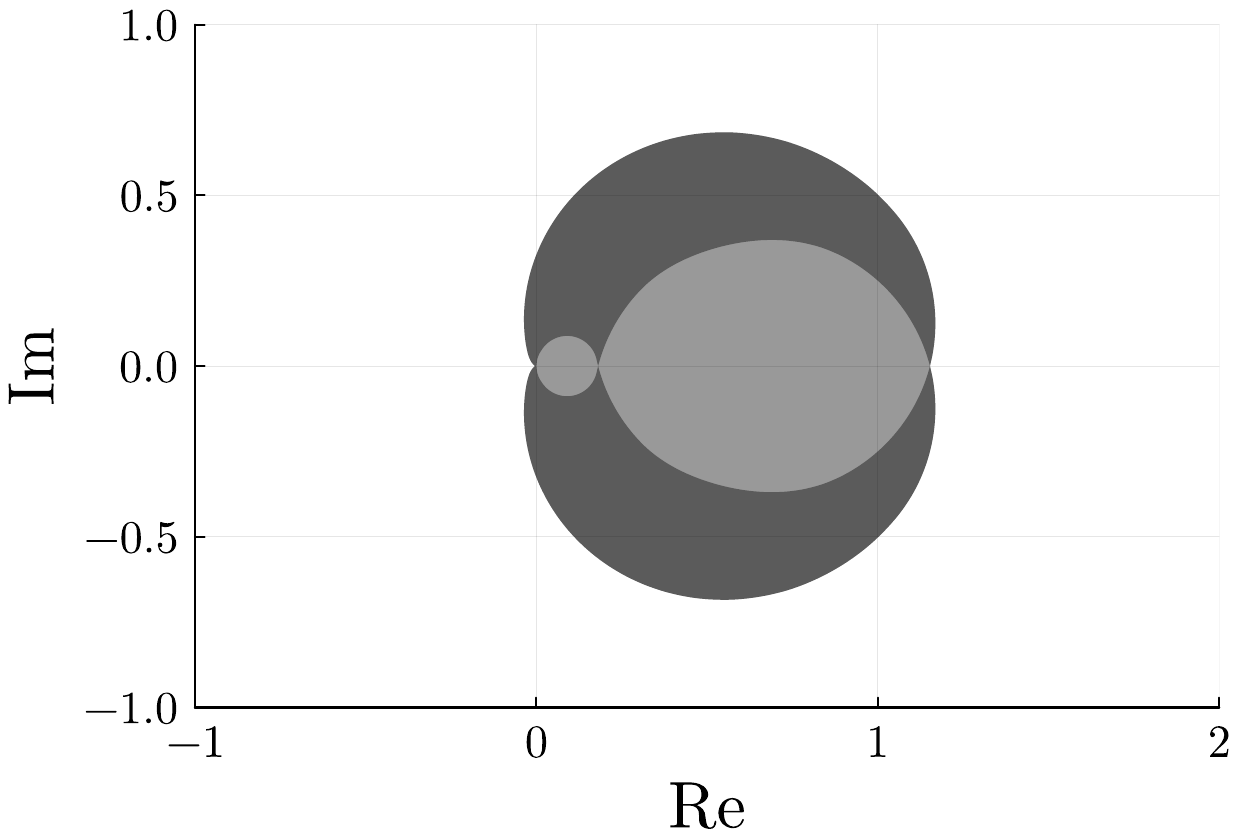}
         \caption{$\SRGe(G_4)$.}
         \label{fig:example_mimo_1}
     \end{subfigure}
     \hfill
     \begin{subfigure}[b]{0.32\linewidth}
         \centering
         \includegraphics[width=\linewidth]{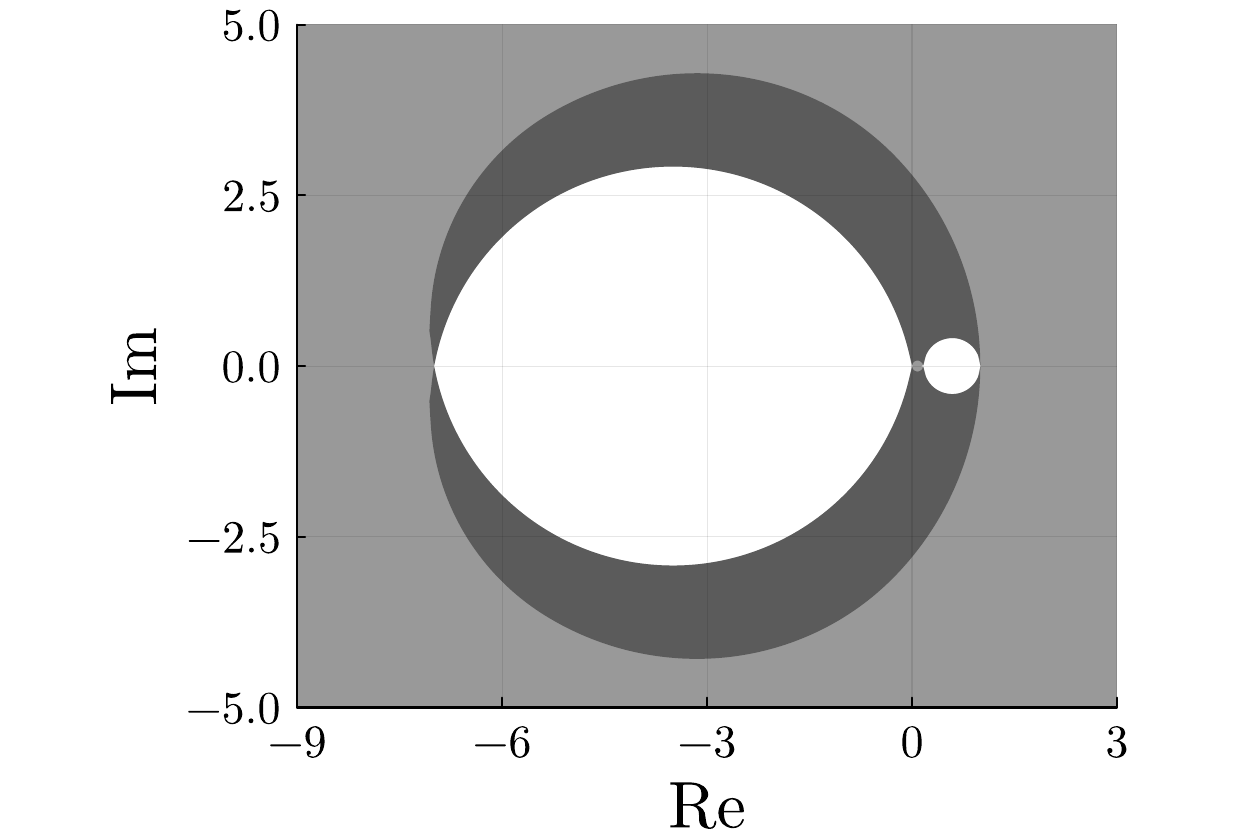}
         \caption{$\SRGe(G_5)$.}
         \label{fig:example_mimo_2}
     \end{subfigure}
     \hfill
     \begin{subfigure}[b]{0.32\linewidth}
         \centering
         \includegraphics[width=\linewidth]{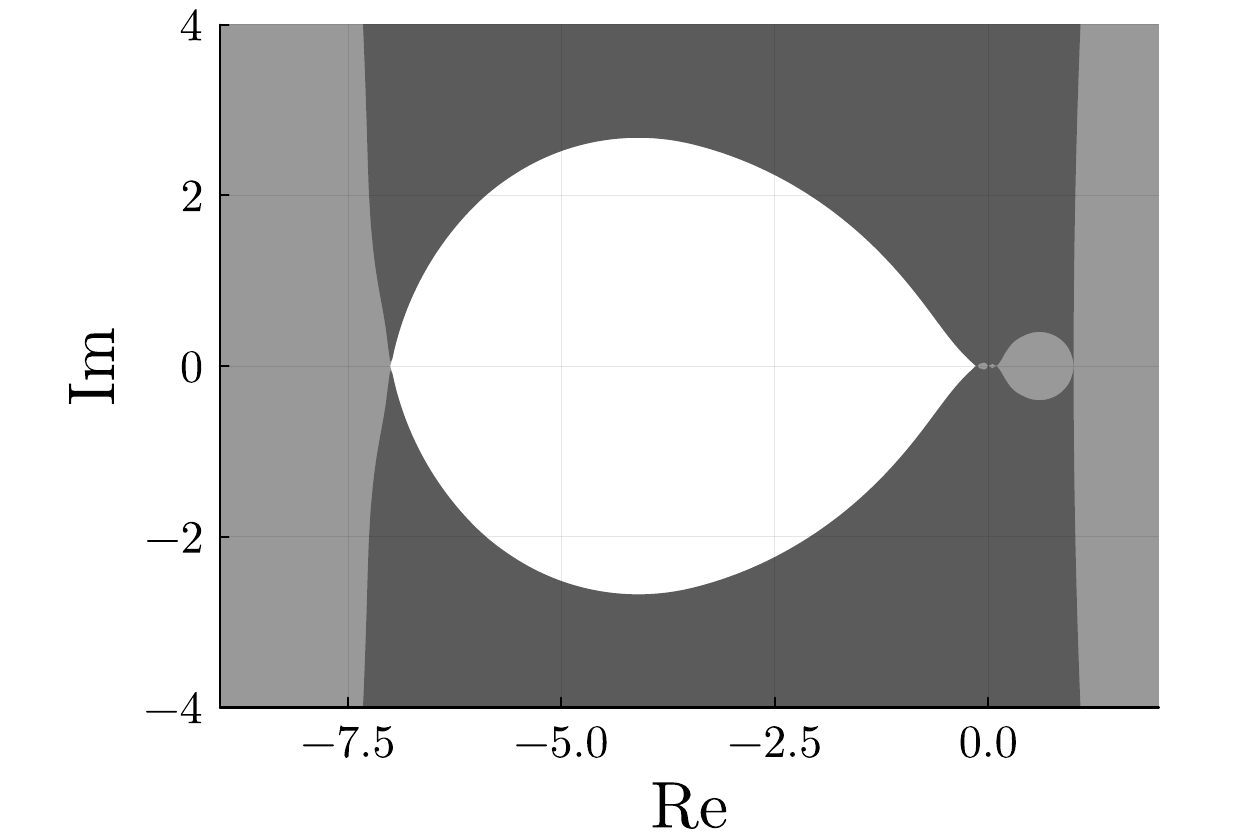}
         \caption{$\SRGe(G_6)$.}
         \label{fig:example_mimo_3}
     \end{subfigure}
    \caption{The hard SRG bound of the systems in Section~\ref{sec:examples} in gray, obtained from~\eqref{eq:hard_srg_G_circle_bound_thm}, where the soft SRG bound is depicted in dark gray and is always contained in the hard SRG bound. Note that for SISO systems, the hard/soft SRG bound is exact.}
    \label{fig:examples_siso}
\end{figure*}

\section{Examples}\label{sec:examples}

We will now demonstrate the use of Theorem~\ref{thm:hard_srg_bound_lti_circles} on various LTI systems. Additionally, we will discuss how the stability of the feedback system in Fig.~\ref{fig:linear_feedback} can be studied using the hard SRG bound of the loop transfer.

\subsection{SISO Examples}

Consider the following SISO transfer functions
\begin{equation*}
    G_1(s) = \frac{1}{s^2+s+1}, \quad G_2(s) = \frac{1}{s(s+1)}, 
\end{equation*}\vspace{-1em}
\begin{equation*}
    G_3(s) = \frac{3}{(s-2)(s/10+1)},
\end{equation*}
for which the hard SRG and soft SRG is shown in Figs.~\ref{fig:example_siso_1}, \ref{fig:example_siso_2}\footnote{The boundary of the soft SRG should be a straight line from the origin to $\pm j \infty$. Since a finite $\alpha$ is used in~\eqref{eq:soft_srg_lti_siso}, the slightly curved boundary is a numerical artifact.} and \ref{fig:example_siso_3}, respectively.

We note that the extended SRG of $-2G_1$ and $G_3$ were computed in~\cite{krebbekxScaledRelativeGraph2025a}, and are the same as the hard SRG in Figs.~\ref{fig:example_siso_1} and~\ref{fig:example_siso_3}. Observe that $\frac{k G_2}{1+k G_2} = \frac{k}{G_2^{-1}+k}$, and the roots of $G_2^{-1} + k = s^2+s+k$ are $s = -\frac{1}{2}(1 \pm \sqrt{1-4k})$, which obey $\mathrm{Re}(s) < 0$ if and only if $k>0$. This means that $\mathcal{N}_{G_2} \cap \R = [0, \infty)$, which shows that Fig.~\ref{fig:example_siso_2} is indeed equal to $\SRG'(G_2)$. 

Note that each of the hard SRGs are separated from $-1$, hence would yield a stable closed loop in Fig.~\ref{fig:linear_feedback}.

\subsection{MIMO Examples}

Consider the following MIMO LTI loop transfers in Fig.~\ref{fig:linear_feedback}
\begin{equation*}
    G_4(s) = \mat{\frac{1}{s+1} & \frac{1}{s+2} \\ \frac{1}{s+4} & \frac{1}{s+3}}, \quad G_5(s) = \mat{ \frac{s+7}{s-1} & \frac{s-5}{(s+2)^2} \\ \frac{1}{(s+4)^3} & \frac{s}{(s+3)^2}},
\end{equation*}
\begin{equation*}
    G_6(s) = \mat{  \frac{s+7}{s-1} & \frac{s-5}{(s+2)^2}  & \frac{1}{s+1.5} \\ \frac{1}{(s+4)^3} & \frac{s}{(s+3)^2} & \frac{1}{s+1.2} \\ \frac{2}{s+1}  & \frac{1}{s+9} & \frac{1}{s}}.
\end{equation*}
Note that $G_4$ is stable, $G_5$ is unstable and $G_6$ contains an integrator. In this case, their hard SRG \emph{bounds} with corresponding soft SRG \emph{bounds} are displayed in Figs.~\ref{fig:example_mimo_1}, \ref{fig:example_mimo_2} and~\ref{fig:example_mimo_3}\footnote{Similar to $G_2$, the boundary of the soft SRG on the left and right should go to $\pm j \infty$ in the $\mathrm{Im}$-direction. Instead, it bends off due to a finite $\alpha$ in~\eqref{eq:soft_srg_lti}.}. Note that in each case, again $-1$ is separated from the hard SRG, hence the system in Fig.~\ref{fig:linear_feedback} would be stable. 

The hard SRG of $G_5$ is particularly interesting since it has two holes, and will be used to show how the hard SRG provides an alternative to the MIMO Nyquist criterion. For $\tilde{G} = k_1+k_2 G_5$, we can visually choose $k_1,k_2$ such that $\SRGe(\tilde{G})$ is separated from $-1$. For instance, since $[-6,0) \nsubseteq \SRGe(G_5)$, we can pick $k_1=0,k_2=1/5$. Another choice is to flip and scale the hard SRG by choosing $k_1=0,k_2 =-2$, which flips the small hole on the right of the origin in Fig.~\ref{fig:example_mimo_2} onto $-1$. Finally, we can simply shift the hard SRG, for instance by $k_1=-1.5, k_2 = 1$, where now the small hole in Fig.~\ref{fig:example_mimo_2} is shifted to the left such that it contains $-1$. In each case~\eqref{eq:k1k2_criterion} holds, which guarantees stability of the closed-loop.

Aside from designing a stable closed loop, one can tune the gains $k_1,k_2$ in $\tilde{G} = k_1 + k_2 G_i$, where $i = 1,\dots, 6$, to achieve a desired separation $\dist(-1, \SRGe(\tilde{G})^{-1}) \geq r >0$ such that $\Gamma(R) \leq 1/r$. Note, however, that the \emph{inverse} of the hard SRG must be used to compute $R$. This is not explored here due to space constraints.

\section{Conclusion}\label{sec:conclusion}

LTI systems are the fundamental building block of nonlinear system analysis using SRGs. In this paper, we provide an algorithm to exactly compute the hard SRG of SISO and square MIMO LTI systems, which may be unstable. The computations are based on the transfer function representation, and has a direct relation to the soft SRG of the LTI system. Then, we proved that in the SISO case, the hard SRG is equal to the extended SRG, where the latter is a union of the soft SRG and the Nyquist stability criterion. 

While the main strength of the SRG framework lies in nonlinear system analysis, we showed how hard SRG analysis also provides an alternative to the MIMO Nyquist criterion and the GNC.

Even though we focus on LTI systems in this paper, the results immediately carry over to hard SRG analysis of nonlinear feedback systems as developed in~\cite{chenSoftHardScaled2025}. Additionally, we expect that the results in this paper will advance the SRG analysis of nonlinear MIMO systems, as developed in~\cite{krebbekxGraphicalAnalysisNonlinear2025}, to include unstable LTI plants in the loop.

\section*{Appendix: proofs}

\begin{proof}[Proof of Proposition~\ref{prop:stable_square_lti_soft_srg}]
    Since $G :L_2^p \to L_2^p$, we know from~\cite{patesScaledRelativeGraph2021} that $\SRG(G)$ is h-convex, and therefore (its closure) can be written as~\eqref{eq:h-convex-set-disk-representation}. It is immediate from Parseval's theorem that $r_\alpha = \underline{\sigma}(G_\alpha)$ and $R_\alpha = \overline{\sigma}(G_\alpha)$ are the largest and smallest values, respectively, such that $r_\alpha \leq \norm{G_\alpha u}/\norm{u} \leq R_\alpha$ holds. This implies that $\mathrm{cl} \, \SRG(G) = \mathcal{G}(\{\underline{\sigma}(G_\alpha)\}, \{\overline{\sigma}(G_\alpha)\})$, proving the result.
\end{proof}

\begin{proof}[Proof of Lemma~\ref{lemma:G_alpha_bound}]
    Since $\iota_{n \leftarrow p}L_2^p[0,T] = \mathcal{U}_{p,T}$ is a closed subspace of the Hilbert space $L_2^n[0,T]$ and $\mat{ I \\ 0_{(n-p) \times p}} u = u$ for all $u \in \mathcal{U}_{p,T}$, we can use the MIMO SRG interconnection rules~\cite[Thm. 2]{krebbekxGraphicalAnalysisNonlinear2025} to conclude 
    \begin{equation}\label{eq:G_alpha_srg_calc_step}
        \SRG(G_\alpha |_T) = \SRG(G |_T) - \alpha
    \end{equation}
    for all $T>0$, where $\SRG(\cdot)$ is the MIMO SRG as defined in~\cite{krebbekxGraphicalAnalysisNonlinear2025}. Then, since $\SRGe(G_\alpha) \subseteq \mathcal{C} \iff \SRG(G_\alpha |_T) \subseteq \mathcal{C}$ for all $T >0$, we can conclude using~\eqref{eq:G_alpha_srg_calc_step} that $\SRG(G|_T) \subseteq \alpha+\mathcal{C}$ for all $T>0$, which proves~\eqref{eq:G_alpha_bound}.  
\end{proof}

\begin{proof}[Proof of Proposition~\ref{prop:maxgain}]
    We consider the four cases proper/improper and stable/unstable separately.

    \emph{Stable and proper:} $G \in RH_\infty^{q \times p}$ implies $G_\alpha \in RH_\infty^{q \times p}$, so from Parseval's identity we know $\norm{G_\alpha u_T} \leq \overline{\sigma}(G_\alpha) \norm{u_T}$ for any $u \in \Lte^p$ and $T>0$. Since $G_\alpha$ is proper (i.e., causal), $(G_\alpha u_T)_T = (G_\alpha u)_T$ hence $\norm{(G_\alpha u)_T}/\norm{u_T} \leq \overline{\sigma}(G_\alpha)$. Since $\overline{\sigma}(G_\alpha) = \norm{G_\alpha}_\infty$, $R_\alpha = \overline{\sigma}(G_\alpha)$ is the smallest possible value of $R_\alpha$ that obeys \eqref{eq:r_R_alpha_defs}.
    
    \emph{Unstable and proper:} now $G(s)$ has an unstable pole (i.e., in the closed RHP), and so has $G_\alpha$. This means that there exists a $u \in L_2^p$ such that $Gu \in \Lte^p \setminus L_2^p$, hence $\lim_{T \to \infty} \norm{(Gu)_T} / \norm{u} = \infty$, so $R_\alpha = \infty$. %
    
    \emph{Improper:} in this case, $G_\alpha(s) = s^k \tilde{G}(s)$ for $k \in \N$ and $\tilde{G}(s)$ obeys $\lim_{|s| \to \infty} \tilde{G}(s) = c \neq 0$, i.e., $G_\alpha$ is a bi-proper system $\tilde{G}$ pre-multiplied by $k$ ideal derivatives. This means that as $\omega \to \infty$, the response to $u(t) = \sin(\omega t)$ grows as $\omega^k$ if $\tilde{G}$ is stable, i.e., $\lim_{\omega \to \infty} \lim_{T \to \infty} \norm{(G u)_T} / \norm{u_T} = \infty$,  and might even contain exponentially growing terms if $\tilde{G} \notin RH_\infty^{q \times p}$. Therefore, $R_\alpha = \infty$.  
\end{proof}

\begin{proof}[Proof of Proposition~\ref{prop:mingain_square}]
    We will invert $G_\alpha(s)$ and apply Proposition~\ref{prop:maxgain}. Note that we can only compute $G_\alpha^{-1}(s)$ if $G_\alpha(s)$ has full normal rank. If $G_\alpha(s)$ does not have full normal rank, then there exists a nonzero $U(s) \in \R^p(s)$ such that $G_\alpha(s) U(s) = 0$, which implies $r_\alpha = 0$. Note that $\underline{\sigma}(G_\alpha)=0$ when $G_\alpha(s)$ does not have full normal rank, hence $r_\alpha = \underline{\sigma}(G_\alpha)$.

    Suppose that $G^{-1}(s)$ exists, then $\underline{\sigma}(G_\alpha) = 1/\overline{\sigma}(G_\alpha^{-1})$. By the Smith-McMillan form, the poles of $G_\alpha(s)$ are the zeros of $G_\alpha^{-1}(s)$, and vice versa~\cite{hespanhaLinearSystemsTheory2009}. We use Proposition~\ref{prop:maxgain} to find an $R_\alpha$ such that $\norm{(G_\alpha^{-1} y )_T} / \norm{y_T} \leq R_\alpha$ $\norm{u_T} \leq \norm{u} = \norm{ G_\alpha^{-1} y_T} \leq R_\alpha \norm{y_T}$, hence $\norm{y_T} / \norm{u_T} \geq 1/R_\alpha = r_\alpha$.

    \emph{Minimum phase:} in this case, $G_\alpha^{-1}(s)$ is stable. If $G_\alpha^{-1}(s)$ is proper, then $G_\alpha^{-1} \in RH_\infty^{p \times p}$ and $r_\alpha = 1/R_\alpha = \underline{\sigma}(G_\alpha)$. If $G_\alpha^{-1}(s)$ is not proper, then $r_\alpha = 1/R_\alpha = 0$ and since $\overline{\sigma}(G) = \infty$ for improper systems, it holds that $r_\alpha =\underline{\sigma}(G_\alpha) = 0$. 

    \emph{Non-minimum phase:} in this case, $G_\alpha^{-1}(s)$ has closed RHP poles, and hence $r_\alpha = 1/R_\alpha = 0$.  
\end{proof}

\begin{proof}[Proof of Theorem~\ref{thm:hard_srg_bound_lti_circles}]
    By Propositions~\ref{prop:maxgain} and~\ref{prop:mingain_square}, we know that~\eqref{eq:r_R_alpha_defs} holds for all $\alpha \in \R$ and $T>0$. This implies that~\eqref{eq:hard_srg_G_circle_bound} must hold for all $\alpha \in \R$, and we can conclude that $\SRGe(G) \subseteq \cap_{\alpha \in \R} (  D_{R_\alpha}(\alpha)  \setminus  D_{r_\alpha}(\alpha) )$. 

    Now we prove equality of the latter. For each $T>0$, we define $r_\alpha^T$ ($R_\alpha^T$) as the largest (smallest) value such that
    \begin{equation}\label{eq:r_R_alpha_T_defs}
        0 \leq r_\alpha^T \leq \norm{(G_\alpha u)_T} / \norm{u_T} \leq R_\alpha^T \leq \infty.
    \end{equation}
    From~\eqref{eq:r_R_alpha_defs} it is clear that $\inf_T r_\alpha^T = r_\alpha$ and $\sup_T R_\alpha^T=R_\alpha$ for all $\alpha \in \R$. We know from~\eqref{eq:h-convex-set-disk-representation} and~\eqref{eq:hard_srg_def} that
    \begin{equation*}
        \SRGe(G) = \bigcup_{T>0} \Big( \bigcap_{\alpha \in \R} \big( D_{R_\alpha^T}(\alpha) \setminus D_{r_\alpha^T}(\alpha) \big) \Big),
    \end{equation*}
    since $G|_T : L_2^p[0,T] \to L_2^p[0,T]$ for all $T>0$, even if $G$ is unstable. Note that
    \begin{equation*}
        \bigcup_{T>0} \Big( \bigcap_{\alpha \in \R} \big( D_{R_\alpha^T}(\alpha) \setminus D_{r_\alpha^T}(\alpha) \big) \Big) = \bigcap_{\alpha \in \R} \big( D_{R_\alpha}(\alpha) \setminus D_{r_\alpha}(\alpha) \big),
    \end{equation*}
    since $D_{R_\alpha^T}(\alpha) \subseteq D_{R_
    \alpha}(\alpha)$ and $D_{r_\alpha}(\alpha) \subseteq D_{r_\alpha^T}(\alpha)$ for all $\alpha \in \R$ and $T>0$, which is identical to the right-hand side of~\eqref{eq:hard_srg_G_circle_bound_thm}, proving the theorem.
\end{proof}

For the proof of Theorem~\ref{thm:hard_and_extended_srg_equivalence}, we need to establish some geometric facts that will be used throughout the proof. 

Let $G(s)\in R(s)$ be a SISO transfer function, then define $\mathrm{Nyquist}(G) = \{G(j \omega) \mid \omega \in \R\}$, where $G(s)$. By~\eqref{eq:soft_srg_lti_siso}, the soft SRG obeys $\SRG(G) \cap \R = \mathrm{Nyquist}(G) \cap \R$, which consists of $N$ finitely many points, hence $(\mathrm{Nyquist}(G) \cap \R )^c = \bigcup_{i=1}^{N+1} I_i$ where $I_i \subset \R$, $I_i \cap I_j = \emptyset$ for $i \neq j$ and the set complement is taken in $\R$. Note that for each $I_i$, by the Nyquist criterion, the feedback system $\frac{kG}{1+kG}$ is either stable or unstable for all $-\frac{1}{k}\in I_i$. Note that stability of $\frac{kG}{1+kG}$ is equivalent to $G+\frac{1}{k}$ being minimum-phase.

Observe that Theorem~\ref{thm:hard_srg_bound_lti_circles} yields different $r_\alpha$ and $R_\alpha$ values than in the soft SRG bound~\eqref{eq:soft_srg_lti_siso} only when $G$ is unstable ($R_\alpha = \infty \geq \overline{\sigma}(G_\alpha)$), or when $G_\alpha$ is non-minimum phase ($r_\alpha = 0\leq \underline{\sigma}(G_\alpha)$). Hence, $\SRG(G) \subseteq \SRGe(G)$. To distinguish between~\eqref{thm:hard_srg_bound_lti_circles} and~\eqref{eq:hard_srg_G_circle_bound_thm}, we denote the hard SRG bound radii as $r_\alpha^\mathrm{e}$ and $R_\alpha^\mathrm{e}$. 

Finally, because of~\eqref{eq:soft_srg_lti_siso} all connected components (CCs) of $\SRG(G)^c$ are path-connected (PC) to $\R$. Because of this, the winding number $N_G(z)$ is the same for all $z$ in a CCs of $\SRG(G)^c$ since $\mathrm{Nyquist}(G) \subseteq \SRG(G)$.

\begin{proof}[Proof of Theorem~\ref{thm:hard_and_extended_srg_equivalence}]
    We consider two separate cases: either $G(s)$ has one or more integrators, or it has none. Our approach is to prove that $\SRGe(G)^c = \SRG'(G)^c$, where the complements are taken in $\C_\infty$. 

    \emph{Case 1: $G$ has integrators:} in this case, $\mathrm{Nyquist}(G)$ is unbounded, hence $R_\alpha = \infty$ for all $\alpha \in \R$ in~\eqref{eq:soft_srg_lti_siso}, which implies $\SRG(G)^c = \bigcup_{\alpha \in \R} D_{r_\alpha}(\alpha)$. Since $\SRG(G) \subseteq \SRGe(G)$, $R_\alpha = \infty$ for all $\alpha \in \R$ in~\eqref{eq:hard_srg_G_circle_bound_thm} as well. 

    Now, $z \in \SRGe(G)^c \iff z \in D_{r_\alpha^\mathrm{e}}(\alpha)$ for some $\alpha \in \R \iff z \in D_{r_\alpha}(\alpha)$ for some $\alpha \in \R$ and $G-\alpha$ is minimum-phase $\iff z \in D_{r_\alpha}(\alpha)$ for some $\alpha \in \R$ and $\frac{kG}{1+kG}$ is stable $\iff z \in \SRG(G)^c$ in a stable CC\footnote{I.e., where $N_G(\tilde{z})+n_\mathrm{p} = 0$ for all $\tilde{z}$ in that CC.} $\iff z \notin \SRG(G)$ and $z \notin \mathcal{N}_G$ (see~\eqref{eq:set_of_encircled_unstable_points}) $\iff z \in \SRG'(G)^c$.

    \emph{Case 2: $G$ has no integrators:} in this case, $\mathrm{Nyquist}(G)$ is bounded, hence $\SRG(G) \subseteq D_R(0)$ for some $R<\infty$.

    \enquote{$\Longrightarrow$}: Note $z \in \SRGe(G)^c \implies z \in D_{r_\alpha^\mathrm{e}}(\alpha)$ or $z \in (D_{R_\alpha^\mathrm{e}}(\alpha))^c$ for some $\alpha \in \R \implies z \in \SRG(G)^c$ and in a stable CC or $\{ z \in (D_{R_\alpha}(\alpha))^c$ for some $\alpha \in \R$ and $G$ is stable$\}$. From the latter statement, since $G$ is stable, there exists an $\epsilon>0$ such that for $k \in (-\epsilon, \epsilon)$, $\frac{kG}{1+kG}$ is stable. Since $\SRG(G) \subseteq D_{R_\alpha}(\alpha)$, $z \in (D_{R_\alpha}(\alpha))^c$ is PC (in $\SRG(G)^c$) to $(-\epsilon, \epsilon)$, hence $z \notin \mathcal{N}_G$, which together with $z \notin \SRG(G)$, implies $z \in \SRG'(G)^c$. 

    \enquote{$\Longleftarrow$}: Note $z \in \SRG'(G)^c \implies z \notin \SRG(G)$ and $z$ in a stable CC $\implies z \in D_{r_\alpha}(\alpha)$ or $z \in (D_{R_\alpha}(\alpha))^c$ for some $\alpha \in \R$ and $z$ in a stable CC $\implies \{z \in D_{r_\alpha^\mathrm{e}}(\alpha) \}$ or the following. Since the winding number $N_G(z)$ is constant on each CC of $\SRG(G)^c$, the fact that $\SRG(G) \subseteq D_{R_\alpha}(\alpha)$ and $z \in (D_{R_\alpha}(\alpha))^c$ implies that $\SRG'(G)$ bounded. The latter implies~\cite{krebbekxScaledRelativeGraph2025} that $G$ is stable. Since $G$ is stable, $R_\alpha = R_\alpha^\mathrm{e}$ for all $\alpha \in \R$. To conclude, we have shown $z \in \SRG'(G)^c \implies \{z \in D_{r_\alpha^\mathrm{e}}(\alpha)$ for some $\alpha \in \R \}$ or $\{z \in (D_{R_\alpha^\mathrm{e}}(\alpha))^c$ for some $\alpha \in \R\} \iff z \in \SRGe(G)^c$.
\end{proof}

\section*{ACKNOWLEDGMENT}

We thank Alejandro W. Sere for useful comments regarding the exactness of the soft SRG bound.

\footnotesize
\bibliographystyle{IEEEtran} 
\bibliography{bibliography} 

\end{document}